\documentclass{article}

\usepackage[light,sfmathbb,nowarning]{kpfonts}

\usepackage{fullpage}
\usepackage{url}
\usepackage[utf8]{inputenc}
\usepackage[small]{caption}
\usepackage{graphicx}
\usepackage{amsmath}
\usepackage{amsthm}
\usepackage{amsfonts}
\usepackage{booktabs}
\usepackage{algorithm}
\usepackage{amssymb,mathtools}
\usepackage{bm}
\usepackage[pdftex,dvipsnames]{xcolor}
\usepackage{authblk}
\usepackage{xspace}
\usepackage{comment}
\usepackage{array}
\usepackage{multirow}
\usepackage{nicefrac}
\usepackage{todonotes}
\usepackage[noend]{algpseudocode}
\usepackage[numbers,sort]{natbib}
\usepackage{wrapfig}
\usepackage{mathdots}
\usepackage{paralist}
\usepackage{thm-restate}
\usepackage{diagbox}
\usepackage[inline]{enumitem}
\usepackage{subcaption}
\usepackage{todonotes}
\usepackage{tikz}
\usepackage{pgfplots}
\pgfplotsset{compat=1.15}

\usepackage[british]{babel}
\usepackage[pagebackref]{hyperref}
\hypersetup{
    colorlinks=true,
    linkcolor=olive,
    citecolor=purple
}

\usepackage[capitalize]{cleveref}
\usepackage{helvet}
\usepackage{courier}
\DeclareCaptionStyle{ruled}{labelfont=normalfont,labelsep=colon,strut=off}
\usepackage{environ}

\newcommand{\repeatproposition}[1]{  \begingroup
  \renewcommand{\theproposition}{\ref{#1}}  \expandafter\expandafter\expandafter\proposition
  \csname repproposition@#1\endcsname
  \endproposition
  \endgroup
  \setcounter{theorem}{\value{theorem}-1}
}

\NewEnviron{repproposition}[1]{  \global\expandafter\xdef\csname repproposition@#1\endcsname{    \unexpanded\expandafter{\BODY}  }  \expandafter\proposition\BODY\unskip\label{#1}\endproposition
}

\newcommand{\repeattheorem}[1]{  \begingroup
  \renewcommand{\thetheorem}{\ref{#1}}  \expandafter\expandafter\expandafter\theorem
  \csname reptheorem@#1\endcsname
  \endtheorem
  \endgroup
  \setcounter{theorem}{\value{theorem}-1}
}

\NewEnviron{reptheorem}[1]{  \global\expandafter\xdef\csname reptheorem@#1\endcsname{    \unexpanded\expandafter{\BODY}  }  \expandafter\theorem\BODY\unskip\label{#1}\endtheorem
}

\newtheorem{theorem}{Theorem}
\newtheorem{proposition}[theorem]{Proposition}

\newtheorem{definition}{Definition}[section]
\newtheorem{example}{Example}[section]

\DeclareMathOperator{\poly}{poly}

\newcommand{\calS}{{\mathcal{S}}}
\newcommand{\calD}{{\mathcal{D}}}

\newcommand{\swap}{{{\mathrm{swap}}}}
\newcommand{\pref}{\succ}

\newcommand{\p}{{{\mathsf{P}}}}
\newcommand{\np}{{{\mathsf{NP}}}}
\newcommand{\wone}{{{\mathsf{W[1]}}}}
\newcommand{\xp}{{{\mathsf{XP}}}}

\hypersetup{pdfauthor   = {Piotr Faliszewski, Krzysztof Sornat, Stanislaw Szufa},
            pdftitle    = {The Complexity of Subelection Isomorphism Problems},
            pdfkeywords = {election, isomorphism, computational complexity, parameterized complexity, approximability}
}

\title{The Complexity of Subelection Isomorphism Problems\thanks{An extended abstract
of this work appears in AAAI 2022~\cite{FaliszewskiSS2021subelection}.}}

\date{}
\author[1]{Piotr Faliszewski\thanks{faliszew@agh.edu.pl}}
\author[1]{Krzysztof Sornat\thanks{sornat@agh.edu.pl}}
\author[1,2]{Stanisław Szufa\thanks{szufa@agh.edu.pl}}
\affil[1]{AGH University, Kraków, Poland}
\affil[2]{Jagiellonian University, Kraków, Poland}

\begin{document}

\maketitle

\begin{abstract}
  We study extensions of the Election Isomorphism problem, focused on
  the existence of isomorphic subelections. Specifically, we propose
  the \textsc{Subelection Isomorphism} and the \textsc{Maximum Common
  Subelection} problems and study their computational complexity and
  approximability. Using our problems in experiments, we
  provide some insights into the nature of several statistical models of elections.
\end{abstract}

\section{Introduction}

We study the computational complexity of several extensions of the \textsc{Election Isomorphism}
problem, recently introduced by~\citet{fal-sko-sli-szu-tal:c:isomorphism}
as an analogue of \textsc{Graph Isomorphism}. While in the latter we are given
two graphs and we ask if they can be made identical by renaming the
vertices, in the former we are given two ordinal elections (i.e.,
elections where each voter ranks the candidates from the most to the
least appealing one) and ask if they can be made identical by renaming
the candidates and reordering the voters. Interestingly, even though
the exact complexity of \textsc{Graph Isomorphism}, as well as of many related
problems, remains elusive~\cite{bab-daw-sch-tor:j:graphi-isomorphism},
\textsc{Election Isomorphism} has a simple polynomial-time
algorithm~\cite{fal-sko-sli-szu-tal:c:isomorphism}. Yet, in many
practical settings perfect isomorphism is too stringent and
approximate variants are necessary. For the case of \textsc{Graph Isomorphism},
researchers considered two types of relaxations:
Either they focused on making a small number of modifications to the
input graphs that make them isomorphic (see, e.g.,
the works of~\citet{arv-koe-kuh-vas:c:approximate-graph-isomorphism}
and~\citet{gro-rat-woe:c:approximate-graph-isomorphism}), or
they sought (maximum) isomorphic subgraphs of the input ones (see,
e.g., the classic paper of~\citet{coo:c:theorem-proving} and the
textbook of~\citet{gar-joh:b:int}; for an overview
focused on applications in cheminformatics we point to
the work of~\citet{ray-wil:j:max-common-subgraph-isomorphism}).
For the case of elections, the former approach was already taken
by~\citet{fal-sko-sli-szu-tal:c:isomorphism} and several other
researchers~\cite{vay-red-fla-cou:t:co-optimal-transport:nips-20,szu-fal-sko-sli-tal:c:map,boe-bre-fal-nie-szu:c:compass}.
Our goal is to explore the latter route.

More specifically, we consider the \textsc{Subelection Isomorphism}
and \textsc{Maximum Common Subelection} families of problems. In the
former, we are given two elections, a smaller and a larger one, and we
ask if it is possible to delete some candidates and voters from the
larger election so that it becomes isomorphic to the smaller one. Put
differently, we ask if the smaller election occurs as a minor in the
larger one. One reason why this problem is interesting is its
connection to restricted preference domains.
For example,
both single-peaked~\cite{bla:b:polsci:committees-elections}
and single-crossing~\citep{mir:j:single-crossing,rob:j:tax}
elections are characterized as those that do not have certain forbidden
minors~\cite{bal-har:j:characterization-single-peaked,bre-che-woe:j:single-crossing}.
We show that \textsc{Subelection Isomorphism} is $\np$-complete and
$\wone$-hard for the parameterization by the size of the smaller
election, which suggests that there are no fast algorithms for the
problem. Fortunately, the characterizations of single-peaked and
single-crossing elections use minors of constant size and,
such elections can be recognized efficiently; indeed, there are
very fast algorithms for these
tasks~\cite{bar-tri:j:stable-matching-from-psychological-model,esc-lan-ozt:c:single-peaked-consistency,elk-fal-sli:c:decloning}.
Our results show that characterizations with
non-constant minors might lead to $\np$-hard recognition problems.

In our second problem, \textsc{Maximum Common Subelection}, we
ask for the largest isomorphic subelections of the two input ones.
Since their size can be used as a (particularly demanding) measure of similarity,
our work is related to those of \citet{fal-sko-sli-szu-tal:c:isomorphism}
and \citet{szu-fal-sko-sli-tal:c:map}. In the former, the authors
define the similarity between elections using variants of the
swap distance and the Spearman footrule (and find these measures to be
intractable), whereas in the latter the authors propose
a polynomial-time computable measure, based on analyzing the frequency
with which candidates appear on particular positions in the votes. While we find that
many of our problems are $\np$-hard (and hard to approximate), we
also find polynomial algorithms, also for practically useful cases.

For both our problems, we consider their ``candidate'' and ``voter'' variants.
For example, in \textsc{Candidate Subelection Isomorphism} we ask if it is possible
to delete candidates from the larger election (but without deleting
any voters) so that it becomes isomorphic with the smaller
one. Similarly, in \textsc{Maximum Common Voter-Subelection} we ask if
we can ensure isomorphism of the two input elections by only deleting
voters (so that at least a given number of voters remains).
In Section~\ref{sec:experiments} we use this latter problem to
evaluate similarity between elections generated from various
statistical cultures.
These results confirm some findings observed by \citet{szu-fal-sko-sli-tal:c:map} and \citet{boe-bre-fal-nie-szu:c:compass} in their ``maps
of elections'' and give a new perspective on some of these statistical cultures.

In the most general variants of our problems, we assume that both
input elections are over different candidate sets and include
different voters. Yet, sometimes it is natural to assume that the
candidates or the voters are the same (for example, in a presidential
election votes collected in two different districts would have the
same candidate sets, but different voters, whereas two consecutive
presidential elections would largely involve the same voters, but not
necessarily the same candidates). We model such scenarios by variants
of our problems where either the matchings between the candidates or
the voters of the input elections are given (indeed, this way we
follow the approach
of~\citet{fal-sko-sli-szu-tal:c:isomorphism}). While one would expect
that having such matchings would make our problems easier, there are
cases where they remain $\np$-hard even with both matchings. This
stands in sharp contrast to the results
of~\citet{fal-sko-sli-szu-tal:c:isomorphism}. For a summary of our
results, see Table~\ref{tab:complexity_results}.

\begin{table*}[t]
\centering
\scalebox{0.857}{
\begin{tabular}{ c | c | c | c | c }
		\toprule
		Problem & no matching & voter-matching & cand.-matching & both matchings\\
		\midrule
		Election Isomorphism & \phantom{N.}$\p$\phantom{-com.} {\small\cite{fal-sko-sli-szu-tal:c:isomorphism}\phantom{Thm.}} & \phantom{N.}$\p$\phantom{-com.} {\small\cite{fal-sko-sli-szu-tal:c:isomorphism}\phantom{Thm.}} & $\p$ {\small\cite{fal-sko-sli-szu-tal:c:isomorphism}\phantom{Thm}} & $\p$ {\small\cite{fal-sko-sli-szu-tal:c:isomorphism}\phantom{Thm}} \\
		\midrule
		Subelection Isomorphism & $\wone$-h.\phantom{.} {\small [Thm.~\ref{thm:sub-elec-np-hard}]} & $\np$-com. {\small [Thm.~\ref{thm:subelection-voter-candidate}]}& $\p$ {\small [Thm.~\ref{thm:poly-algo}]}& $\p$ {\small [Thm.~\ref{thm:poly-algo}]} \\

		Cand.-Subelection Isomorphism & $\np$-com. {\small [Prop.~\ref{prop:can-subelection}]}& $\np$-com. {\small [Thm.~\ref{thm:subelection-voter-candidate}]}& $\p$ {\small [Thm.~\ref{thm:poly-algo}]}& $\p$ {\small [Thm.~\ref{thm:poly-algo}]}\\

		Voter-Subelection Isomorphism & \phantom{N}$\p$\phantom{-com.} {\small [Thm.~\ref{thm:poly-algo}]}& \phantom{N}$\p$\phantom{-com.} {\small [Thm.~\ref{thm:poly-algo}]}& $\p$ {\small [Thm.~\ref{thm:poly-algo}]}& $\p$ {\small [Thm.~\ref{thm:poly-algo}]}\\
		\midrule
		Max. Common Subelection & $\wone$-h.\phantom{.} {\small [Prop.~\ref{prop:mcs}]}& $\np$-com. {\small [Prop.~\ref{prop:mcs}]}& $\np$-com. {\small [Prop.~\ref{prop:mcs}]}& \phantom{[1]}$\np$-com. {\small [Prop.~\ref{prop:mcs}]}\\

		Max. Common Cand.-Subelection & $\np$-com. {\small [Prop.~\ref{prop:mccs-cm-vm}]}& $\np$-com. {\small [Prop.~\ref{prop:mccs-cm-vm}]}& $\np$-com. {\small [Prop.~\ref{prop:mccs-cm-vm}]}& $\wone$-com. {\small [Thm.~\ref{thm:common-cand-both-np-hard}]}\\

		Max. Common Voter-Subelection & \phantom{.}$\p$\phantom{-com.} {\small [Thm.~\ref{thm:poly-algo}]} & \phantom{.}$\p$\phantom{-com.} {\small [Thm.~\ref{thm:poly-algo}]}& \phantom{.}$\p$\phantom{-com.} {\small [Thm.~\ref{thm:poly-algo}]}& \phantom{[1].}$\p$\phantom{-com.} {\small [Thm.~\ref{thm:poly-algo}]}\\
		\bottomrule
\end{tabular}
}
\caption{\label{tab:complexity_results} An overview of our results;
    those for \textsc{Election Isomorphism} are due to~\citet{fal-sko-sli-szu-tal:c:isomorphism}.
    $\wone$-hardness holds with respect to the size of the smaller election or a common subelection.
    Indicated $\wone$-hard problems are also $\np$-hard.
  }
\end{table*}

\section{Preliminaries}\label{sec:prelim}

For a positive integer $k$, we write $[k]$ to denote the set
$\{1,\dots,k\}$, and by $S_k$ we refer to the set of all permutation
over $[k]$. Given a graph $G$, we write $V(G)$ to refer to its set of
vertices and $E(G)$ to refer to its set of edges.

\paragraph{Elections.}
An election is a pair $E=(C,V)$ that consists of a set
$C = \{c_1, \ldots, c_m\}$ of candidates and a collection
$V = (v_1, \ldots, v_n)$ of voters. Each voter $v\in V$ has a
preference order, i.e., a ranking of the candidates from the most to
the least appreciated one, denoted as $\pref_v$. Given two candidates
$c_i, c_j \in C$, we write $c_i \pref_v c_j$ (or, equivalently,
$v \colon c_i \pref c_j$) to denote that $v$ prefers $c_i$ to
$c_j$. We extend this notation to more than two candidates in a
natural way. For example, we write
$v \colon c_1 \pref c_2 \pref \cdots \pref c_m$ to indicate that voter
$v$ likes $c_1$ best, then $c_2$, and so on, until $c_m$. If we put
some set $S$ of candidates in such a description of a preference
order, then we mean listing its members in some arbitrary (but fixed, global)
order. Including $\overleftarrow{S}$ means listing the members of $S$
in the reverse of this order.
We often refer to the preference orders as either the votes or the
voters, but the exact meaning will always be clear from the context.
By the \textit{size} of an election, we mean the number of candidates multiplied by the number of voters.
Occasionally we discuss single-peaked elections.

\begin{definition}[\citet{bla:b:polsci:committees-elections}]
  Let $C$ be a set of candidates, and let $\lhd$ be
  a linear order over $C$ (referred to as the societal axis).
  We say that a vote $v$ is single-peaked
  with respect to~$\lhd$ if for each $k \in [m]$,
  the top $k$ candidates in $v$ form an interval within $\lhd$.
  An election is single-peaked if there is a societal axis for which
  all its votes are single-peaked.
\end{definition}

Given elections $E = (C,V)$ and $E' = (C',V')$, we say that $E'$ is a
\emph{subelection} of $E$ if $C'$ is a subset of $C$ and $V'$ can be
obtained from $V$ by deleting some voters and restricting the remaining
ones to the candidates from~$C'$.
We say that $E'$ is a \emph{voter subelection} of $E$
if we can obtain it by only deleting voters from $E$,
and that $E'$ is a \emph{candidate subelection} of~$E$
if we can obtain it from $E$ by only deleting candidates.

\paragraph{Election Isomorphism.}
Let $E$ be an election with candidate set $C = \{c_1, \ldots, c_m\}$
and voter collection $V = (v_1, \ldots, v_n)$. Further, let $D$ be
another set of $m$ candidates and let $\sigma \colon C \rightarrow D$
be a bijection between $C$ and $D$. For each voter $v \in V$, by
$\sigma(v)$ we mean a voter with the same preference order as $v$,
except that each candidate $c$ is replaced with $\sigma(c)$. By
$\sigma(V)$ we mean voter collection $(\sigma(v_1), \ldots, \sigma(v_n))$.
Similarly, given a permutation $\pi \in S_n$, by $\pi(V)$ we mean
$(v_{\pi(1)}, \ldots, v_{\pi(n)})$.

Two elections are \emph{isomorphic} if it is possible to rename their
candidates and reorder their voters so that they become
identical~\cite{fal-sko-sli-szu-tal:c:isomorphism}. Formally,
elections $(C_1,V_1)$ and $(C_2,V_2)$, are isomorphic if
$|C_1| = |C_2|$, $|V_1| = |V_2|$, and there is a bijection
$\sigma \colon C_1 \rightarrow C_2$ and a permutation
$\pi \in S_{|V_1|}$ such that $(\sigma(C_1),\sigma(\pi(V_1))) = (C_2,V_2)$.
We refer to $\sigma$ as the candidate matching and to $\pi$ as the
voter matching. In the \textsc{Election Isomorphism} problem we are
given two elections and we ask if they are isomorphic.

\paragraph{Computational Complexity.} We assume familiarity with
(parameterized) computational complexity theory; for background, we point the
readers to the textbooks of~\citet{pap:b:complexity}
and~\citet{cyg-fom-kow-lok-mar-pil-pil-sau:b:fpt}.
Most of our intractability proofs follow by reductions from the
\textsc{Clique} problem. An instance of \textsc{Clique} consists of a
graph $G$ and a nonnegative integer $k$, and we ask if $G$ contains
$k$ vertices that are all connected to each other. \textsc{Clique} is
well-known to be both $\np$-complete and $\wone$-complete, for the
parameterization by~$k$~\cite{downey1995fixed}.
Additionally, we provide some lower-bounds based on the Exponential Time Hypothesis (ETH),
which is a popular conjecture on solving the CNF-SAT problem.
For a formal statement see, e.g.,
Conjecture 14.1 in~\cite{cyg-fom-kow-lok-mar-pil-pil-sau:b:fpt}.
Specifically, in our lower-bound proofs we use a consequence of the ETH which
says that there is no $|V(G)|^{o(k)}$-time algorithm for \textsc{Clique}~\cite{ChenHKX06}.
As all the problems that we study can easily be seen to belong to $\np$,
in our $\np$-completeness proofs we only give hardness arguments.

\section{Variants of the Isomorphism Problem}
We introduce two extensions of the \textsc{Election Isomorphism} problem,
inspired by \textsc{Subgraph Isomorphism} and
\textsc{Maximum Common Subgraph}. In the former, we are given two
elections and we ask if the smaller one is isomorphic to a subelection
of the larger one. That is, we ask if we can remove some
candidates and voters from the larger election to make the two
elections isomorphic.

\begin{definition}
  An instance of \textsc{Subelection Isomorphism} consists of two
  elections, $E_1 = (C_1, V_1)$ and $E_2 = (C_2,V_2)$, such that
  $|C_1| \leq |C_2|$ and $|V_1| \leq |V_2|$. We ask if there is a
  subelection $E'$ of $E_2$ isomorphic to $E_1$.
\end{definition}
The \textsc{Voter-Subelection Isomorphism} problem is defined in the
same way, except that we require $E'$ to be a voter subelection of
$E_2$. Similarly, in \textsc{Candidate-Subelection Isomorphism} we
require $E'$ to be a candidate subelection. We often abbreviate the
name of the latter problem to \textsc{Cand.-Subelection Isomorphism}.

\begin{example}\label{ex:1}
  Consider elections $E = (C,V)$ and $F = (D, U)$, where $C=\{a,b,c\}$,
  $D = \{x,y,z,w\}$, $V=(v_1,v_2,v_3)$ and $U=(u_1,u_2,u_3)$, with 
  preference orders:
  \begin{align*}
    &v_1 \colon a \pref b \pref c,
    &u_1 \colon w \pref x \pref y \pref z,& \\
    &v_2 \colon b \pref a \pref c,
    &u_2 \colon y \pref w \pref x \pref z,& \\
    &v_3 \colon c \pref b \pref a,
    &u_3 \colon z \pref w \pref y \pref x.
  \end{align*}
  If we remove candidate $w$ from $(D,U)$, then we find that the
  resulting elections are isomorphic (to see this, it suffices to match
  voters $v_1,v_2,v_3$ with $u_1,u_2,u_3$, respectively, and
  candidates $a,b,c$ with $x,y,z$). Thus $E$ is isomorphic to a
  (candidate) subelection of $F$ and, so, $(E,F)$ is a
  \emph{yes}-instance of \textsc{(Cand.-)Subelection Isomorphism}.
\end{example}

In the \textsc{Maximum Common Subelection} problem
we seek the largest isomorphic subelections of two given ones.
We often abbreviate \textsc{Maximum} as \textsc{Max.}

\begin{definition}
  An instance of \textsc{Max. Common Subelection} consists of two
  elections, $E_1 = (C_1,V_1)$ and $E_2 = (C_2,V_2)$, and a
  positive integer $t$. We ask if there is a subelection $E'_1$
  of $E_1$ and a subelection $E'_2$ of $E_2$ such that $E'_1$ and
  $E'_2$ are isomorphic and the size of $E'_1$ (or, equivalently, the
  size of $E'_2$) is at least $t$.
\end{definition}

Analogously to the case of \textsc{Subelection Isomor\-phism},
we also consider the \textsc{Max. Common Cand.-Sub\-elec\-tion} and
\textsc{Max. Common Voter-Subelection} problems. In the former, $E'_1$
and $E'_2$ must be candidate subelections
and in the latter they need to be voter subelections
(thus in the former problem $E_1$ and $E_2$ must have the
same numbers of voters, and in the latter $E_1$ and $E_2$ must have
the same numbers of candidates).

For each of the above-defined problems we consider its variant
with or without the candidate or voter matching. Specifically, the
variants defined above are \emph{with no matchings}. Variants
\emph{with candidate matching} include a bijection~$\sigma$ that
matches (some of) the candidates in one election to (some of) those
in the other (in case of \textsc{Subelection Isomorphism}
and its variants, all the candidates from the smaller election must be
matched to those in the larger one;
in case of \textsc{Max. Common Subelection} there are no such requirements).
Then we ask for an isomorphism between respective subelections that agrees with
$\sigma$. In particular, this means that none of the unmatched
candidates remain in the considered subelections (another
interpretation is to assume that both input election have the
same candidate sets).
\begin{example}
  Consider elections $(C,V)$ and $(D, U)$ from Example~\ref{ex:1}, and
  a matching $\sigma$ such that $\sigma(a)=x, \sigma(b)=w$, where $c$,
  $y$, and $z$ are unmatched.
  After applying it and dropping the unmatched candidates,
  the votes in the first election become
  $v_1 \colon x \pref w$, $v_2 \colon w \pref x$, and
  $v_3 \colon w \pref x$, whereas all the voters in the second election
  have preference order $w \pref x$. Thus, this instance of
  \textsc{Max. Common Subelection with Candidate Matching} has
  isomorphic subelections, respecting the matching $\sigma$, of size $2 \cdot 2 = 4$.
\end{example}
Variants \emph{with voter matching} are defined similarly: We are
given a matching between (some of) the voters from one election and
(some of) the voters from the other (and, again, for
\textsc{Subelection Isomorphism} and its variants, each voter from the
smaller election is matched to some voter in the larger one). The
sought-after isomorphism must respect this matching (again, this means
that we can disregard the unmatched voters).

Variants \emph{with both matchings} include both the matching between
the candidates and the matching between the voters (note that these
variants are not trivial because we still need to decide who to
delete). By writing \emph{all four matching cases} we mean the four
just-described variants of a given problem.

Finally, we note that each variant of \textsc{Max. Common Subelection}
is at least as computationally difficult as its corresponding variant
of \textsc{Subelection Isomorphism}.

\begin{repproposition}{M reduces to S}
  \label{prop:reduction}
  Let $M$ be a variant of \textsc{Max. Common Subelection} and let $S$
  be a corresponding variant of \textsc{Subelection Isomorphism}. We
  have that $S$ many-one reduces to $M$ in polynomial time.
\end{repproposition}

\section{Computational Complexity Analysis}\label{sec:comp-complexity}

In this section, we present our complexity results. While in most cases
we obtain intractability (see Table~\ref{tab:complexity_results} for a
summary of our results), we find that all our problems focused on
voter subelections are solvable in polynomial time, and having
candidate matchings leads to polynomial-time algorithm for all the
variants of \textsc{Subelection Isomorphism}.
Some proofs were relegated to the Appendix~\ref{appendix:proofs} for readability.

All our polynomial-time results rely on the trick
used by~\citet{fal-sko-sli-szu-tal:c:isomorphism} for the
case of \textsc{Election Isomorphism}. The idea is to guess a pair
of (matched) voters and use them to derive the candidate matching.

\begin{reptheorem}{Polynomial-time algorithm results}
\label{thm:poly-algo}
  \textsc{Voter-Subelection Isomorphism} and
  \textsc{Max. Common Voter-Subelection} are in $\p$ for all four matching
  cases. \textsc{Subelection Isomorphism},
  \textsc{Cand.-Subelection Isomorphism} are in $\p$ for cases
  with candidate matchings.
\end{reptheorem}

\subsection{Intractability of Subelection Isomorphism}

Next we show computational hardness for all the remaining variants of our problems.
In this section we consider \textsc{Subelection Isomorphism}.

\begin{theorem}\label{thm:sub-elec-np-hard}
  \textsc{Subelection Isomorphism} is $\np$-com\-plete and
  $\wone$-hard with respect to the size of the smaller election.
\end{theorem}

\begin{proof}
  Before we describe our reduction, we first provide a method for
  transforming a graph into an election: For a graph $H$, we let $E_H$
  be an election whose candidate set consists of the vertices of $H$
  and two special candidates, $\alpha_H$ and $\beta_H$, and whose
  voters correspond to the edges of $H$. Specifically, for each edge
  $e = \{x,y\} \in E(H)$ we have four voters, $v^1_e, \ldots v^4_e$,
  with preference orders:
  \begin{align*}
    v^1_e \colon& x \pref y \pref \alpha_H \pref \beta_H \pref V(H) \setminus \{x,y\}, \\
    v^2_e \colon& x \pref y \pref \beta_H \pref \alpha_H \pref V(H) \setminus \{x,y\}, \\
    v^3_e \colon& y \pref x \pref \alpha_H \pref \beta_H \pref V(H) \setminus \{x,y\}, \\
    v^4_e \colon& y \pref x \pref \beta_H \pref \alpha_H \pref V(H) \setminus \{x,y\}.
  \end{align*}
  We give a reduction from \textsc{Clique}. Given an instance $(G,k)$
  of \textsc{Clique}, where $G$ has at least $k$ vertices and
  $\binom{k}{2}$ edges,
  we let $K$ be a size-$k$ complete graph and we form an instance
  $(E_K, E_G)$ of \textsc{Subelection Isomorphism}. The reduction
  runs in polynomial time and it remains to show correctness.

  First, let us assume that $G$ has a size-$k$ clique. Let $X$ be the
  set of its vertices and let $Y$ be the set of its edges. We form a
  subelection $E'$ of $E_G$ by removing all the candidates that are
  not in $X \cup \{\alpha_G, \beta_G\}$ and removing all the voters that
  do not correspond to the edges from $Y$. One can verify that $E'$ and
  $E_K$ are, indeed, isomorphic.

  Second, let us assume that $E_K$ is isomorphic to some subelection
  $E'$ of $E_G$. We will show that this implies that $G$ has a
  size-$k$ clique. First, we claim that $E'$ includes both
  $\alpha_G$ and $\beta_G$. To see why this is so, consider the
  following two cases:
  \begin{enumerate}
  \item If $E'$ contained exactly one of $\alpha_G, \beta_G$, then this
    candidate would appear in every vote in $E'$ among the top three
    positions. Yet, in $E_K$ there is no candidate with this property
    so $E'$ and $E_K$ would not be isomorphic.
  \item If $E'$ contained neither $\alpha_G$ nor $\beta_G$ then every
    vote in $E'$ would rank some vertex candidates $z$ and $w$ on
    positions three and four
    (to be able to match $\alpha_K$ and $\beta_K$ to them).
    However, by the construction of $E_G$,
    either in every vote from $E'$ we would have $z \pref w$ or in
    every vote from $E'$ we would have $w \pref z$. Since in $E_K$
    half of the voters rank the candidates from positions three and
    four in the opposite way, $E'$ and $E_K$ would not be isomorphic.
  \end{enumerate}
  Thus $\alpha_G$ and $\beta_G$ are included in $E'$.
  Moreover $\alpha_G$ and $\beta_G$ are matched with
  $\alpha_K$ and $\beta_K$ because they are the only
  candidates from $E_G$ that can appear on positions three and four
  in every vote in $E'$ but possibly in different order.
  As a consequence, for each vote $v$ from $E_G$ that appears in $E'$,
  the candidate set of $E'$ must include
  the two candidates from $V(G)$ that $v$ ranks
  on top (if it were not the case, then $E'$ would contain a
  candidate---either $\alpha_G$ or $\beta_G$---that appeared in all
  the votes within the top four positions and in some vote within top
  two positions; yet, $E_K$ does not have such a candidate).
  This means that for each edge $e \in E(G)$, if $E'$ contains some
  voter $v^i_e$ for $i \in [4]$, then it also contains the other
  voters corresponding to $e$ (otherwise $E'$ and $E_K$ would not be
  isomorphic).
  Because the number of voters in $E_K$ is $4\binom{k}{2}$,
  hence the number of distinct chosen (in $E'$) corresponding edges
  from $G$ equals $\binom{k}{2}$.
  As said before, for each such chosen edge, we also choose two
  corresponding vertices as candidates.
  It means that the number of chosen candidates
  (except $\alpha_G$ and $\beta_G$) is between $k$ and $2\binom{k}{2}$.
  However, the number of candidates in $E_K$ except $\alpha_K$ and
  $\beta_K$ is $k$, therefore we conclude that chosen vertex-candidates
  form a size-$k$ clique in $G$.
  This completes the proof of $\np$-hardness.

  To show $\wone$-hardness, note that the number of candidates and
  voters in the smaller election equals $k+2$ and $4\binom{k}{2}$
  respectively, hence the size of the smaller election is a function of
  parameter $k$ for which \textsc{Clique} is $\wone$-hard.
\end{proof}

The above reduction shows something stronger than the theorem claims.
Indeed, assuming ETH one cannot compute a solution faster than a straightforward brute-force approach.

\begin{repproposition}{Subelection Isomorphism ETH-lowerbound}
  \textsc{Subelection Isomorphism} has an $O^*(m^k)$-time algorithm,
  where $k$ is the number of candidates in the smaller election and
  $m$ is the number of candidates in the larger election
  (hence $\xp(k)$).
  Moreover, assuming ETH, there is no $O^*(m^{o(k)})$-time algorithm.
\end{repproposition}

As a consequence of the above $\xp$ algorithm,
even testing if a fairly small, constant-sized,
election is isomorphic to a subelection of some bigger one may require
a polynomial-time algorithm with impractically large
exponents. Luckily, in practice this is not always the case. For
example, single-peaked elections are characterized as exactly those
that do not have subelections isomorphic to certain two elections of
sizes $8$ and $9$~\cite{bal-har:j:characterization-single-peaked}, but
there is an algorithm for deciding if a given election is
single-peaked that runs in time $O(nm)$, where $m$ is the number of
candidates and $n$ is the number of
voters~\cite{esc-lan-ozt:c:single-peaked-consistency}.
For single-crossing elections, such a characterization uses subelections
of sizes up to 18, but there is a recognition algorithm that runs in
time $O(n^2 +nm\log m)$~\cite{doi-fal:j:unidimensional}, which can
even be improved using appropriate data structures.

Next we consider \textsc{Cand.-Subelection Isomorphism}.
In this problem both elections have the same number of voters and
we ask if we can delete candidates from the one that has more,
so that they become isomorphic.
We first show that this problem is $\np$-complete for the
case where the voter matching is given (which also proves the same
result for \textsc{Subelection Isomorphism with Voter Matching}) and
next we describe how this proof can be adapted to the variant without
any matchings (the variant with candidate matching is in $\p$ and was
considered in the preceding section).

\begin{reptheorem}{Cand-Subelection with voter matching hardness}\label{thm:subelection-voter-candidate}
  \textsc{Subelection Isomorphism with Voter Matching} and
  \textsc{Cand.-Subelection Isomorphism with Voter Matching} are
  $\np$-complete.
\end{reptheorem}

\textsc{Cand.-Subelection Isomorphism}
remains $\np$-com\-plete also without the voter matching.
By doubling the voters and using a few extra candidates we ensure that
only the intended voter matching is possible.

\begin{repproposition}{Cand.-Subelection Isomorphism problem}
\label{prop:can-subelection}
\textsc{Cand.-Subelection Isomorphism} is $\np$-complete.
\end{repproposition}

\subsection{Intractability of Max. Common Subelection}

Perhaps the most surprising result regarding
\textsc{Max. Common Subelection} is that it is
$\np$-complete even when both matchings are given.
The surprise stems from the fact that all generalizations
of \textsc{Election Isomorphism} considered
by~\citet{fal-sko-sli-szu-tal:c:isomorphism} are solvable in
polynomial-time in this setting. We first show this result for
candidate subelections.

\begin{reptheorem}{Common-Cand with both matching hardness}\label{thm:common-cand-both-np-hard}
\textsc{Max. Common Cand.-Subelection with Both Matchings} is $\np$-complete and $\wone$-complete with respect to the candidate set size of isomorphic candidate subelections.
\end{reptheorem}

\begin{proof}
  We give a reduction from the \textsc{Clique} problem
  which idea is to decode the adjacency matrix of a given graph
  by a pair of elections with both matchings defined.
  Missing edges in a graph we decode as a conflict on candidates
  ordering within matched voters.

  Formally, given an instance $(G,k)$ of \textsc{Clique}, we form
  two elections, $E_1 = (C,V_1)$ and $E_2 = (C,V_2)$,
  where $C = V(G)$.
  Since we need to provide an instance with a candidate matching,
  we simply specify both elections over the same candidate set.
  Without loss of generality, we assume that $V(G) = \{1, \ldots, n\}$.
  For each $x \in V(G)$ we define the neighborhood of $x$ in $G$
  as $N(x) = \{y \in V(G): \{x,y\} \in E(G)\}$ and the set of
  non-neighbors as $M(x) = V(G) \setminus \{N(x) \cup \{x\}\}$.

  For each vertex $x \in V(G)$ we define two matched voters,
  $v^1_x$ in $E_1$ and $v^2_x$ in $E_2$, defined as follows:
  \begin{align*}
    &v^1_x \colon M(x) \pref x \pref N(x),&
    &\text{and}
    &v^2_x \colon x \pref M(x) \pref N(x).&
  \end{align*}
  We ask if $E_1$ and $E_2$ have isomorphic candidate subelections
  that contain at least~$k$ candidates each.
  Intuitively, in a solution to the problem one have to remove
  either $x$ or all vertices from $M(x)$.
  It is a direct definition of a clique: Either $x$ is not in
  a clique or all its non-neighbors are not in a clique.
  It is clear that the reduction can be computed in polynomial time
  and it remains to show its correctness.

  First let us assume that $G$ has a size-$k$ clique. Let $K$ be the
  set of this clique's vertices. We form elections $E'_1$ and $E'_2$
  by restricting $E_1$ and $E_2$ to the candidates from $K$.
  To verify that $E'_1$ and $E'_2$ are isomorphic via the given
  matchings, let us consider an arbitrary pair of matched voters
  $v^1_x$ and $v^2_x$.
  If $x$ is not included in $K$ then preference orders of $v^1_x$
  and $v^2_x$ restricted to $K$ are identical.
  Indeed, removing even only $x$ from the set of candidates
  makes $v^1_x$ and $v^2_x$ identical.
  Otherwise, if $x$ is in $K$ then $K \cap M(x) = \emptyset$ as
  $K$ is a clique. Therefore, removing $M(x)$ from the set of
  candidates makes $v^1_x$ and $v^2_x$ identical.

  For the other direction, let us assume that there are subelections
  $E'_1$ and $E'_2$ of $E_1$ and $E_2$, respectively, each with
  candidate set $K$, such that $|K| \geq k$ and $E'_1$ and $E'_2$ are
  isomorphic via the given matchings. It must be the case that the
  vertices from $K$ form a clique because if $K$ contained two
  vertices $x$ and $y$ that were not connected by an edge, then votes
  $v^1_x$ and $v^2_x$ would not be identical.
  Indeed, then we would have $y \pref x$ in $v^1_x$ and
  $x \pref y$ in $v^2_x$, respectively, when restricted to candidates
  from $K$. This completes the proof of $\np$-hardness.

  To show $\wone$-hardness, note that the required number of
  candidates in isomorphic candidate subelections is equal to the
  parameter $k$ for which \textsc{Clique} is $\wone$-hard.

  For membership in $\wone$ we show a reduction to \textsc{Clique} with equal value of the parameters.
  Its idea is to create a complete graph in which vertices correspond to candidates and then remove edges based on conflicted voters.
  More precisely, for every two matched voters $v$ and $u$ and every two candidates $x$ and $y$ such that $x \pref_v y$ and $y \pref_u x$, we remove an edge $\{x,y\}$ from the graph.
  It means that both vertices corresponding to $x$ and $y$ cannot stay in a solution of \textsc{Clique}.
  Indeed, as we cannot delete voters, the only way to resolve this conflict is to remove a candidate $x$ or $y$ (or both).
  A formal reduction is available in the Appendix~\ref{appendix:proofs}.
\end{proof}

The above reduction can be used to show strong hardness
results which transfer from \textsc{Clique}.
In particular, a brute-force algorithm is essentially the best possible for exact computation.
What is more, a trivial approximation algorithm which returns a constant size solution (hence the approximation ratio is $O(m)$) is also essentially optimal.

\begin{repproposition}{Max-Cand-Both XP}\label{prop:com-cand-both-xp}
  \textsc{Max. Common Cand.-Subelection with Both Matchings} has an $O^*(m^k)$-time algorithm,
  where $k$ is the number of candidates in isomorphic candidate subelections and $m$ is the number of candidates in the input (hence $\xp(k)$).
  Moreover, assuming ETH, there is no $O^*(m^{o(k)})$-time algorithm.
\end{repproposition}

\begin{repproposition}{Max-Cand-Both approximation}\label{prop:com-cand-both-apx}
  \textsc{Max. Common Cand.-Subelection with Both Matchings}
  has an $t/c$-approximation algorithm for any constant $c \geq 1$,
  where $t$ is the maximum number of candidates in isomorphic candidate subelections.
  Moreover, approximating the problem within $t^{1-\epsilon}$ factor,
  for every $\epsilon > 0$, is $\np$-hard.
\end{repproposition}

All the remaining variants of \textsc{Max. Common Cand.-Subelection}
also are $\np$-complete. The proofs either follow by applying
Proposition~\ref{prop:reduction} or by introducing candidates that
implement a required voter matching. In the latter case,
$\wone$-hardness does not follow from this reduction as we introduce
dummy candidates that have to be included in a solution, but their
number is not a function of the \textsc{Clique} parameter (clique
size).

\begin{repproposition}{Max. Common Cand.-Subelection problem with
    candidate matching}
    \label{prop:mccs-cm-vm}
    \textsc{Max. Common Cand.-Subelection} is $\np$-complete and so are its
    variants with a given candidate matching and with a given voter matching.
\end{repproposition}

Similarly to all four matching cases of
the \textsc{Max. Common Cand.-Subelection}, all four matching cases of the
\textsc{Max. Common Subelection} also are $\np$-complete.

\begin{repproposition}{All 4 MCS}
\label{prop:mcs}
All four matching cases of the \textsc{Max. Common Subelection} are $\np$-complete.
\end{repproposition}

\section{Experiments}\label{sec:experiments}

In this section we use the \textsc{Max. Common Voter-Sub\-elec\-tion} problem to analyze
similarity between elections generated from various
statistical models. While \textsc{Max. Common Voter-Subelection} has
a polynomial-time algorithm, it is too slow for our purposes.
Thus we have expressed it as an integer linear program (ILP) and
we were solving it using the CPLEX ILP solver.
A formal ILP formulation is available in the Appendix~\ref{appendix:ilp}.
The source code used for the experiments is available in a GitHub
repository\footnote{\url{https://github.com/Project-PRAGMA/Subelections-AAAI-2022}}.

We stress that we could have used other problems from the \textsc{Max. Common Subelection} family in this section.
We chose \textsc{Max. Common Voter-Subelection} because its outcomes are particularly easy to interpret.

Our findings are similar to those of
\citet{szu-fal-sko-sli-tal:c:map}, but our claims of similarity
between statistical cultures are stronger, whereas our dissimilarity
claims are weaker.

\subsection{Statistical Models of Elections}

We consider the following models for generating elections:

\begin{figure*}[t]
     \centering
     \scalebox{0.925}{
        {\includegraphics[width=5.95cm]{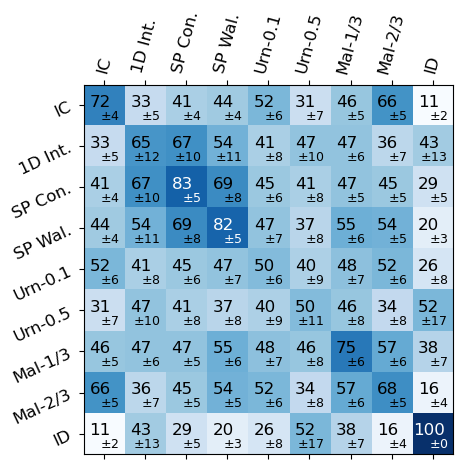}}\quad
        {\includegraphics[width=5.07cm]{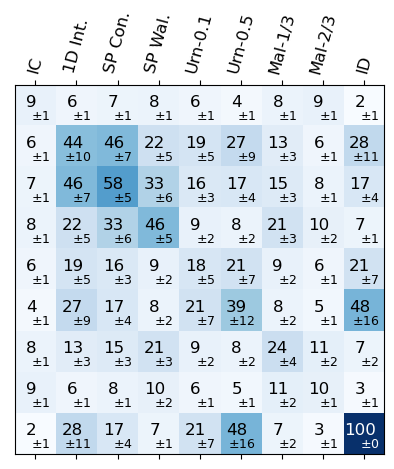}}\quad
        {\includegraphics[width=5.95cm]{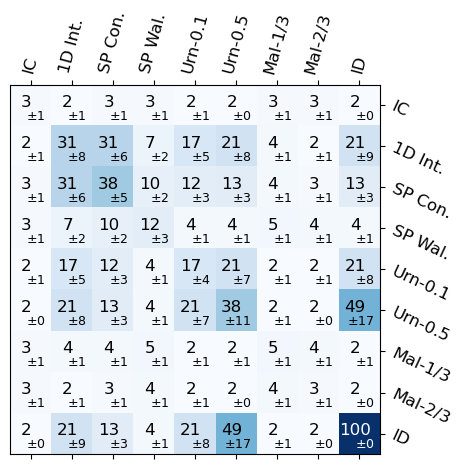}}
    }
    \caption{\label{fig:experiment} The large numbers denote the rounded \% of matched votes for \textsc{Max. Common Voter-Subelection}. The small numbers denote the rounded standard deviation. In the left/center/right matrix there are results for 4/7/10 candidates.}
\end{figure*}

\begin{figure}[t]
    \centering
    \scalebox{1.0}{
        {\includegraphics[width=7cm]{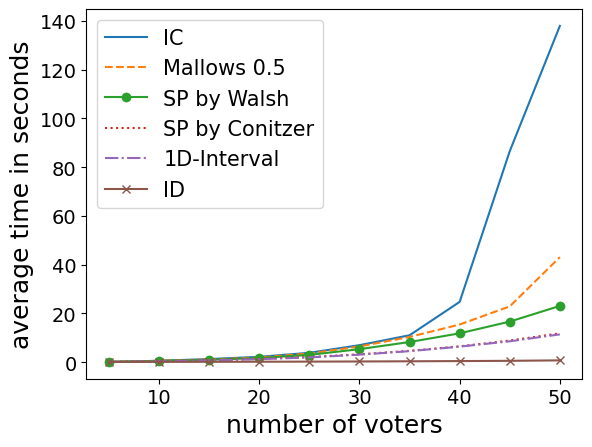}}\quad
        {\includegraphics[width=7cm]{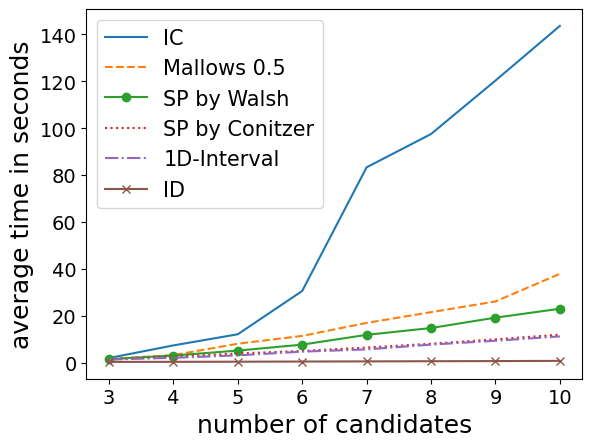}}\quad
    }
    \caption{\label{fig:experiment_time} Average time needed to find the maximum common voter subelections. With the fixed number of candidates (left), and fixed number of voters (right).}
\end{figure}

\begin{description}

\item[Identity.] Under the Identity model (ID), we choose a single
  preference order uniformly at random and then all the generated
  votes are equal to it.

\item[Impartial Culture.] Under Impartial Culture (IC), we
  generate each preference order uniformly at random.

\item[Pólya-Eggenberger Urn Model.] The Pólya-Eggenberger Urn
  Model~\cite{berg1985paradox} is parameterized by a nonnegative value
  $\alpha \in [0,\infty)$, which specifies the degree of correlation
  between the votes~\cite{mcc-sli:j:similarity-rules}. We start with
  an urn containing exactly one copy of each possible preference
  order; then we generate the votes one by one, by drawing a
  preference order from the urn (which we assume to be the generated
  vote) and replacing it with $\alpha \cdot m!$ duplicates, where $m$
  is the number of candidates in the election.

\item[Mallows Model.] The Mallows Model~\cite{mal:j:mallows} is
  parameterized by a value $\phi \in [0,1]$ and a central preference
  order $v_c$, which we choose uniformly at random. The probability of
  generating a preference order $v$ is proportional to
  $\phi^{\swap(v_c, v)}$, where $\swap(v_c, v)$ is the minimum number
  of swaps of adjacent candidates needed to transform~$v$ into $v_c$.
  In our experiments we do not set the value of $\phi$ directly, but
  we use the parameterization by $\mathrm{norm\hbox{-}}\phi$ proposed
  by~\citet{boe-bre-fal-nie-szu:c:compass}
  (specifically, they used
  $\mathrm{rel}\hbox{-}\phi = 0.5\cdot\mathrm{norm}\hbox{-}\phi$).
  It works as follows: For a given $\mathrm{norm\hbox{-}}\phi$ value
  and a given number $m$ of candidates in the election to be
  generated, we choose value $\phi$ so that for a generated
  vote $v$ the expected value of $\swap(v_c,v)$ is equal to
  $\mathrm{norm\hbox{-}}\phi$ times half the total number of swaps
  possible in preference orders over $m$ candidates.
  Thus using $\mathrm{norm}\hbox{-}\phi = 1$ means generating votes according to
  the IC model, using $\mathrm{norm}\hbox{-}\phi = 0$ means using the
  identity model, and using $\mathrm{norm}\hbox{-}\phi = 0.5$ means
  using a model that in a certain formal sense is exactly between
  these two extremes.
  See the work of~\citet{lu-bou:j:sampling-mallows} for an effective
  sampling algorithm.

\item[1D Interval Model.] The candidates and the voters are
  points drawn uniformly at random from a unit interval.
  Each voter~$v$ ranks the candidates with respect to increasing Euclidean
  distances of their points from that of $v$.

\end{description}

We also use the models of
\citet{wal:t:generate-sp} and \citet{con:j:eliciting-singlepeaked}
that generate single-peaked elections (it is also well-known that all 1D
Interval elections are single-peaked).
Given a societal axis, these models work as follows:
\begin{description}

\item[Walsh Model.] Each single-peaked vote, for a given axis, is drawn with equal probability
  (we use the sampling algorithm of \citet{wal:t:generate-sp}).

\item[Conitzer Model.] Under the Conitzer model, to generate a vote we
  start by choosing the top candidate uniformly at random. Then we
  keep on extending the vote with either of the two candidates right
  next to the already selected one(s) on the axis, depending on a
  coin toss.
\end{description}

\subsection{Results and Analysis}
We consider elections with $4$, $7$, or $10$ candidates and with
$50$ voters.
For each scenario and
each two of the above-described models, we have generated $1000$ pairs
of elections (for urn elections, we used $\alpha \in \{0.1, 0.5\}$ and
for the Mallows model, we used
$\mathrm{norm}$-$\phi \in \{\nicefrac{1}{3}, \nicefrac{2}{3}\}$). For each pair of
models, we recorded the average number of voters in the maximum common
voter subelections (normalized by fifty, i.e., the number of voters in
the original elections), as well as the standard deviation of this
value. We show our numerical results in Figure~\ref{fig:experiment}
(each cell corresponds to a pair of models; the number in the top-left
corner is the average, and the one in the bottom-right corner is the
standard deviation). Note that the matrices in
Figure~\ref{fig:experiment} are symmetric (results for models $A$ and
$B$ are the same as for models $B$ and $A$).

For the case with four candidates,
we see that the level of similarity between elections from
various models is quite high and drops sharply as the number of candidates increases.
This shows that for experiments with very few candidates it is not as
relevant to consider very different election models, but for
more using diverse models is justified.

The above notwithstanding, some models remain
similar even for 7 or 10 candidates. This is particularly
visible for the case of single-peaked elections. The 1D-Interval model
remains very similar to the Conitzer model, and the Walsh model is
quite similar to these two for up to $7$ candidates, but for 10
candidates starts to stand out (the result for 10 candidates is also
visible in the experiments of~\citet{szu-fal-sko-sli-tal:c:map};
the similarity for fewer candidates is a new finding).

We also note that the urn models remain relatively similar to each
other (and to the 1D-Interval and Conitzer models) for all numbers of
candidates, but this is not the case for the Mallows models. One
explanation for this is that the urn model proceeds by copying some of
the votes already present in the election, whereas the Mallows model
generates votes by perturbing the central one. The former leads to
more identical votes in an election. Indeed, to verify this it
suffices to consider the ``ID'' column (or row) of the matrix: The
similarity to the identity elections simply shows how often the most
frequent vote appears in elections from a given model. For 10
candidates, urn elections with $\alpha \in \{0.1, 0.5\}$ have, on
average, $21\pm8\%$ and $49\pm17\%$ identical votes, respectively. For
Mallows elections, this value drops to around $2\%$ (in our
setting, this means 1 or 2 voters, on average).

Finally, we consider the diagonals of the matrices in
Figure~\ref{fig:experiment}, which show self-similarity of our models.
Intuitively, the larger are these values, the fewer elections of a
given type one needs in an experiment.
Single-peaked elections stand out here for all numbers of candidates,
whereas the urn models become more prominent for larger candidate sets.

We have also analyzed the average running time that CPLEX\footnote{We ran CPLEX
on a single thread (Intel(R) Xeon(R) Platinum 8280 CPU @ 2.70GH)
of a 448 thread machine with 6TB of RAM.} needed to
find the maximum common voter subelections.
We focus on IC, Identity, Walsh model,
Conitzer model, Mallows model with norm-$\phi=0.5$, and
1D-Interval. First, we generated $500$ pairs of elections from each
model with $10$ candidates and $5,10,\dots,45,50$ voters, and
calculated the average time needed to find the maximum common voter
subelections. Second, we fixed the number of voters to $50$ and
generated elections with $3,4,\dots,9,10$ candidates. The results are
presented in Figure~\ref{fig:experiment_time}. As we increase the
number of voters, the time seems to increase exponentially. We observe
large differences between the models, with IC being by far the
slowest. Conitzer model and Walsh model are significantly different
from one another, even though both of them generate single-peaked
elections. Moreover, the fact that 1D-Interval and Conitzer models
need on average the same amount of time confirms their similarity.

\section{Conclusions}
We have shown that variants of \textsc{Election Isomorphism} that are
based on considering subelections are largely intractable but,
nonetheless, some of them can be solved in polynomial-time.
Indeed, we have used the polynomial-time solvable
\textsc{Maximum Common Candidate Subelection} problem to analyze similarity
between various different models of generating random elections.

In Section~\ref{sec:comp-complexity} we have classified variants of the problem as either
belonging to $\p$ or being $\np$-complete (and some being $\wone$-hard).
For some variants of our problems we have shown strong inapproximability results and matching approximation algorithms.
However, it would also be desirable to consider the parameterized complexity of these problem for the remaining variants.

It would be interesting to consider
variants of our problems where instead of requiring identical
preference orders among matched voters, we might ask for similar ones
(e.g., within a given swap distance). Another interesting research
direction is to repeat our experiments on real-life elections.

\section*{Acknowledgments}
We would like to thank the anonymous reviewers for their valuable comments.
Stanis{\l}aw Szufa was supported by the National Science Centre, Poland (NCN) grant No 2018/29/N/ST6/01303.
This project has received funding from the European
Research Council (ERC) under the European Union’s Horizon 2020
research and innovation programme (grant agreement No 101002854).

\begin{center}
    \includegraphics[width=3cm]{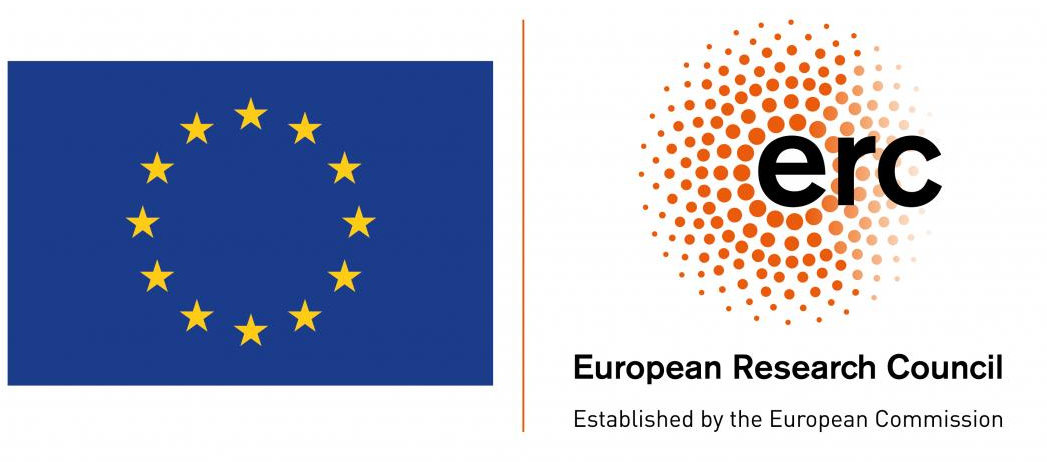}
\end{center}

\bibliographystyle{plainnat}
\bibliography{bib.bib}

\clearpage

\appendix

\section{Missing Proofs}\label{appendix:proofs}

Below we provide the missing proofs from the main text.

\repeatproposition{M reduces to S}

\begin{proof}
  We are given a problem $M$ which is a variant of \textsc{Max. Common
  Subelection}, and problem $S$, which is an analogous variant of
  \textsc{Subelection Isomorphism} (so if the former only allows
  deleting candidates, then so does the latter, etc.). We want to show
  that $S \leq_m^p M$, i.e., that $S$ many-one reduces to $M$ in
  polynomial time. Let $I_S = (E_1,E_2)$ be an instance of $S$, where
  $E_1$ is the smaller election. We form an instance
  $I_M = (E_1,E_2,t)$ of $M$, which uses the very same elections and
  where $t$ is set to be the size of $E_1$. This means that in $I_M$
  we cannot perform any operations on election $E_1$, because that
  would decrease its size below $t$. Hence, we can only perform
  operations on $E_2$, so the situation is the same as in the $S$
  problem and in the $I_S$ instance.
\end{proof}

\repeattheorem{Polynomial-time algorithm results}

\begin{proof}
  We first give an algorithm for \textsc{Max. Common Voter-Subelection}.
  Let $E = (C, V)$ and $F = (D,U)$ be our input elections and let $t$
  be the desired size of their isomorphic subelections. Since we are
  looking for a voter subelection, without loss of generality we may
  assume that $|C| = |D|$ (and we write $m$ to denote the number of candidates
  in each set). For each voter $v \in V$ and each voter
  $u \in U$ we perform the following algorithm:
  \begin{enumerate}
  \item Denoting the preference orders of $v$ and $u$ as
    $c_1 \pref_v c_2 \pref_v \cdots \pref_v c_m$ and
    $d_1 \pref_u d_2 \pref_u \cdots \pref_u d_m$, respectively, we
    form a bijection $\sigma \colon C \rightarrow D$ such that for
    each $c_i \in C$ we have $\sigma(c_i) = d_i$.
  \item We form a bipartite graph where the voters from $V$ form one
    set of vertices, the voters from $U$ form the other set of
    vertices, and there is an edge between voters $v' \in V$ and
    $u' \in U$ if $\sigma(v') = u'$.
  \item We compute the maximum cardinality matching in this graph and
    form subelections that consist of the matched voters. We accept if
    their size is at least $t$.
  \end{enumerate}
  If the algorithm does not accept for any choice of $v$ and $u$, we
  reject.

Very similar algorithms also work for the variants of \textsc{Max. Common
Voter-Subelection} with either one or both of the matchings: If we are
given the candidate matching, then we can omit the first step in the
enumerated algorithm above, and if we are given a voter matching then
instead of trying all pairs of voters $v$ and $u$ it suffices to try
all voters from the first election and obtain the other one via the
matching. Analogous algorithms also work for
\textsc{Voter-Subelection Isomorphism} (for all four matching cases)
and for all the other variants of \textsc{Subelection Isomorphism},
provided that the candidate matching is given.
\end{proof}

\repeatproposition{Subelection Isomorphism ETH-lowerbound}
\begin{proof}
  The $\xp(k)$ algorithm is a straightforward brute-force approach.
  For a fixed order of candidates in the smaller election,
  we create a matching with some $k$ candidates in the larger election.
  We try at most $m(m-1)\dots(m-k+1) \leq m^k$ cases.
  For a given matching, we solve a problem
  \textsc{Subelection Isomorphism with Candidate Matching} using
  a polynomial-time algorithm from Theorem~\ref{thm:poly-algo}.
  Hence the total running time is $O^*(m^k)$.

  For the hardness part, recall that ETH
  implies that there is no $|V|^{o(K)}$-time algorithm for finding
  size-$K$ clique on a graph with vertex set $V$~\cite{ChenHKX06}.
  In our reduction in Theorem~\ref{thm:sub-elec-np-hard} we have that
  the number of candidates in the larger election is $m = |V|+2$ and
  the number of candidates in the smaller election is $k = K+2$.
  Therefore, using an $O^*(m^{o(k)})$-time algorithm for
  \textsc{Subelection Isomorphism} we could solve \textsc{Clique} in
  time $m^{o(k)} \cdot \poly(k^3+m^3) = (|V|+2)^{o(K+2)} \cdot
  \poly(|V|) = |V|^{o(K)}$. This contradicts ETH.
\end{proof}

\repeattheorem{Cand-Subelection with voter matching hardness}
\begin{proof}
  It suffices to consider \textsc{Cand.-Subelection Isomorphism with
  Voter Matching}. We give a reduction from the \textsc{Exact Cover
  by 3-Sets} problem \textsc{(X3C)}. An instance of \textsc{X3C}
  consists of a set $X = \{x_1, \ldots, x_m\}$ of elements and a
  family $\calS = \{S_1, \ldots, S_n\}$ of three-element subsets of
  $X$. We ask if $\calS$ contains a subfamily $\calS'$ such that each
  element from $X$ belongs to exactly one set from $S'$. Given such
  an instance, we form two elections, $E_1$ and $E_2$, as follows.

  Election $E_1$ will be our smaller election and $E_2$ will be the
  larger one. We let $X$ be the candidate set for election $E_1$,
  whereas to form the candidate set of $E_2$ we proceed as follows.
  For each set $S_t = \{x_i, x_j, x_k\} \in \calS$, we introduce
  candidates $s_{i,t}$, $s_{j,t}$, and $s_{k,t}$. Intuitively, if some
  candidate $s_{u,t}$ will remain in a subelection isomorphic to
  $E_1$, then we will interpret this fact as saying that element $x_u$
  is covered by set $S_t$; this will, of course, require introducing
  appropriate consistency gadgets to ensure that $S_t$ also covers its
  other members. For each $i \in [m]$, we let $P_i$ be the set of all
  the candidates of the form $s_{i,t}$, where $t$ belongs to $[n]$ (in
  other words, each $P_i$ contains those candidates from $E_2$ that
  are associated with element $x_i$).

  \begin{example}
    Let $X=\{x_1,x_2,x_3,x_4,x_5,x_6\}$ and
    $\calS=\{S_1, \ldots, S_4\}$, where $S_1 = \{x_1,x_2,x_5\}$,
    $S_2 = \{x_1,x_3,x_6\}$, $S_3 = \{x_2,x_3,x_5\}$, and
    $S_4 = \{x_3, x_4,x_6\} \}$. Then
    $P_1 = \{s_{1,1}, s_{1,2}\}, P_2 = \{s_{2,1}, s_{2,3}\}, P_3 =
    \{s_{3,2}, $
    $ s_{3,3}, s_{3,4}\}, P_4 = \{s_{4,4}\}, P_5 = \{s_{5,1},
    s_{5,3}\}, P_6 = \{s_{6,2}, s_{6,4}\}$.
  \end{example}

  We denote the voter collection of election $E_1$ as
  $V = (v, v', v_1, v'_1, \ldots, v_{n}, v'_n)$ and the voter
  collection of $E_2$ as
  $U = (u, u', u_1, u'_1, \ldots, u_{n}, u'_n)$. Voters $v$ and $v'$
  are matched to voters $u$ and $u'$, respectively, and for each
  $i \in [2n]$, $v_i$ is matched to $u_i$ and $v'_i$ is matched to
  $u'_i$. The preference orders of the first two pairs of voters are:
  \begin{align*}
    v  &\colon x_1 \succ x_2 \succ \dots \succ x_m, &
    u  &\colon P_1 \succ P_2 \succ \dots \succ P_m, \\
    v' &\colon x_1 \succ x_2 \succ \dots \succ x_m, &
    u' &\colon
      \overleftarrow{P}_1 \succ \overleftarrow{P}_2 \succ \dots \succ
      \overleftarrow{P}_m.
  \end{align*}
  Next, for each $t \in [n]$ such that $S_t = \{x_i,x_j,x_k\}$,
  $i < j < k$, we let the preference orders of $v_{t}$, $v'_{t}$ and
  their counterparts from $U$ be as follows (by writing ``$\cdots$'' in
  the votes from $E_1$ we mean listing the candidates from
  $X \setminus \{x_i,x_j,x_k\}$ in the order of increasing indices, and for
  the voters from $E_2$ by ``$\cdots$'' we mean order
  $P_1 \pref P_2 \pref \cdots \pref P_m$ with $P_i$, $P_j$, and $P_k$
  removed):
  \begin{align*}
    v_{t} & \colon x_{i} \pref x_{j} \pref x_{k} \pref \cdots,\\
    u_{t} & \colon s_{i,t} \pref s_{j,t} \pref s_{k,t}
        \pref P_{i} \setminus \{s_{i,t}\} \pref P_{j} \setminus \{s_{j,t}\} \pref P_{k} \setminus \{s_{k,t}\}
        \pref \cdots, \\
    v'_{t} & \colon x_{k} \pref x_{j} \pref x_{i} \pref \cdots,\\
    u'_{t} & \colon s_{k,t} \pref s_{j,t} \pref s_{i,t}
        \pref
        P_{k} \setminus \{s_{k,t}\} \pref P_{j} \setminus \{s_{j,t}\} \pref
        P_{i} \setminus \{s_{i,t}\} \pref \cdots.
  \end{align*}
  This finishes our construction. It is clear that it is
  polynomial-time computable and it remains to show that it is
  correct.

  Let us assume that we have a \emph{yes}-instance of \textsc{X3C},
  that is, there is a family $\calS'$ of sets from $\calS$ such that
  each element from $X$ belongs to exactly one set from $\calS'$. We
  form a subelection $E'$ of $E_2$ by deleting all the candidates
  $s_{i,t}$ except for those for whom set $S_t$ belongs to $\calS'$.
  Then, let $\sigma$ be a function such that for each $x_i \in X$ we
  have $\sigma(x_i) = s_{i,t}$, where $S_t$ is a set from $\calS'$
  that contains $x_i$. Together with our voter matching, $\sigma$
  witnesses that $E_1$ and $E'$ are isomorphic.

  For the other direction, let us assume that $E_2$ has a subelection
  $E'$ that is isomorphic to $E_1$ and let $C'$ be its candidate set.
  First, we claim that for each $i \in [m]$ exactly one candidate from
  $P_i$ is included in $C'$. Indeed, if it were not the case, then $u$
  and $u'$ would not have identical preference orders, as required by
  the fact that they are matched to $v$ and $v'$. Second, we note
  that for each $i \in [m]$ the candidate matching that witnesses our
  isomorphism must match $x_i$ with the single candidate in
  $P_i \cap C'$.
  Finally, we claim that if some candidate $s_{i,t}$ is included in
  $C'$, where $S_t = \{x_i,x_j,x_k\}$, $i < j < k$, then candidates
  $s_{j,t}$ and $s_{k,t}$ are included in $C'$ as well. Indeed, if
  $s_{j,t}$ were not included in $C'$, then $u_t$ and $u'_t$ would
  rank the member of $P_i \cap C'$ and the member of $P_j \cap C'$ in
  the same order, whereas $v_j$ and $v'_j$ rank $x_i$ and $x_j$ in
  opposite orders (and, by the second observation, $x_i$ and $x_j$ are
  matched to the member of $P_i \cap C'$ and $P_j \cap C'$,
  respectively). The same argument applies to $x_{k,t}$.
  We say that a set $S_t = \{x_i, x_j, x_k\} \in \calS$ is selected by
  $C'$ if the candidates $s_{i,t}$, $s_{j,t}$, and $s_{k,t}$ belong to
  $C'$. By the above reasoning, we see that exactly $n/3$ sets are
  selected and that they form an exact cover of~$X$.
\end{proof}

\repeatproposition{Cand.-Subelection Isomorphism problem}

\begin{proof}
  We give a reduction from \textsc{Cand.~Subelection Isomorphism with
  Voter Matching}. Let the input instance be $(E_1,E_2)$, where the
  smaller election, $E_1$, has voter collection $(v_1, \ldots, v_n)$
  and the larger election, $E_2$, has voter collection
  $(u_1, \ldots, u_n)$. Further, for each $i \in [n]$ voter $v_i$ is
  matched with voter $u_i$. We form elections $E'_1$ and $E'_2$ in
  the following way. The candidate set of $E'_1$ is the same as that
  of $E_1$ except that it also includes candidates from the set
  $D = \{d_1, \ldots, d_{2n}\}$. Similarly, $E'_2$ contains the same
  candidates as $E_2$ plus the candidates from the set
  $F = \{f_1, \ldots, f_{2n}\}$. The voter collections of $E'_1$ and
  $E'_2$ are, respectively, $(v'_1, v''_1, \ldots, v'_n,v''_n)$ and
  $(u'_1, u''_1, \ldots, u'_n, u''_n)$. For each $i \in [n]$ these
  voters have the following preference orders (by writing $[v_i]$ or
  $[u_i]$ in a preference order we mean inserting the preference order
  of voter $v_i$ or $u_i$ in a given place:
  \begin{align*}
    v'_{i} \colon & d_{2i-1} \pref d_1 \pref \cdots \pref d_{2n} \pref [v_i], \\
    u'_{i} \colon & f_{2i-1} \pref f_1 \pref \cdots \pref f_{2n} \pref [u_i], \\
    v''_{i} \colon & d_{2i} \pref d_1 \pref \cdots \pref d_{2n} \pref [v_i], \\
    u''_{i} \colon & f_{2i} \pref f_1 \pref \cdots \pref f_{2n} \pref [u_i].
  \end{align*}
  We claim that $E'_1$ is isomorphic to a candidates subelection of
  $E'_2$ if and only if $E_1$ is isomorphic to a candidate subelection
  of $E_2$ with the given voter matching. In one direction this is
  clear: If $E_1$ is isomorphic to a subelection of $E_2$ with a given
  voter matching, then it suffices to use the same voter matching for
  the case of $E'_1$ and $E'_2$, and the same candidate matching,
  extended with matching each candidate $d_i$ to $f_i$.

  Next, let us assume that $E'_1$ is isomorphic to some subelection
  $E'$ of $E'_2$. By a simple counting argument, we note that $E'$
  must contain some candidates not in $F$. Further, we also note that
  it must contain all members of $F$. Indeed, each voter in $E'_1$ has
  a different candidate on top and this would not be the case in $E'$
  if it did not include all members of $F$ (if $E'$ did not include
  any members of $F$ then this would hold for each two votes $u'_i$
  and $u''_i$, and if $E'$ contained some members of $F$ but not all
  of them, then this would hold because each voter in $E'$
  would rank some member of $F$ on top, but there would be fewer
  members of $F$ than voters in the election).

  As a consequence, every candidate matching $\sigma$ that witnesses
  isomorphism between $E'_1$ and $E'$ matches some member of $D$ to
  some member of $F$. Further, we claim that for each $i \in [2n]$,
  $\sigma(d_i) = f_i$. For the sake of contradiction, let us assume
  that this is not the case and consider some $i \in [2n-1]$ for which
  there are $j$ and $k$ such that $\sigma(d_i) = f_j$,
  $\sigma(d_{i+1}) = f_k$ and $j > k$ (such $i$, $j$, $k$ must exist
  under our assumption). However, in $E'_1$, all but one voter rank
  $d_i$ ahead of $d_{i+1}$, whereas in $E'$ all but one voter rank
  $\sigma(d_{i+1})$ ahead of $\sigma(d_i)$. Thus $\sigma$ cannot
  witness isomorphism between $E'_1$ and $E'$.

  Finally, since for each $i \in [2n]$ we have that $d_i$ is matched
  to $f_i$, it also must be the case that for each $j \in [n]$ voters
  $v'_j$ and $v''_j$ are matched to $u'_j$ and $u''_j$, respectively
  (indeed, $v'_j$ is the only voter who ranks $d_{2j-1}$ on top, and
  $u'_j$ is the only voter who ranks $f_{2j-1}$ on top; the same
  argument works for the other pair of voters). As a consequence, we
  have that $E_1$ is isomorphic to a subelection of $E_2$ under the
  voter matching that for each $i \in [n]$ matches $v_i$ to $u_i$.
\end{proof}

\repeattheorem{Common-Cand with both matching hardness}

\begin{proof}[Proof of belonging to $\wone$]
  The proof of the other claims is included in the main text.
  Here, we show a reduction to \textsc{Clique} with equal value of the parameters which implies inclusion in $\wone$.
  Its idea is presented in the main text, so here we begin with its formal proof directly.

  Let $E_1 = (C,V_1)$ and $E_2 = (C,V_2)$ be our input elections and let $k$ be the number of candidates in maximum isomorphic candidate subelections (since we are in the ``with candidate matching'' regime, we take the candidate sets to be equal).
  Let $m = |C|$, $n = |V_1| = |V_2|$ (since we cannot remove the voters).

  We create an instance $(G,k)$ of \textsc{Clique} as follows.
  We define $G$ as having vertices corresponding to candidates, i.e., $V(G) = C$.
  We construct the set of edges by beginning from a complete graph and removing some of them as follows.
  For every two matched voters $v$ and $u$ and every two candidates $x$ and $y$ such that $x \pref_v y$ and $y \pref_u x$, we remove an edge $\{x,y\}$ from the graph.
  It is clear that the reduction can be computed in polynomial time and both parameters have the same value.
  It remains to show its correctness.

  First, let us assume that there are subelections
  $E'_1$ and $E'_2$ of $E_1$ and $E_2$, respectively, each with
  candidate set $K$, such that $|K| \geq k$ and $E'_1$ and $E'_2$ are
  isomorphic via the given matchings.
  It must be the case that the vertices from $K$ form a clique.
  Indeed, if $K$ contained two vertices $x$ and $y$ that were not connected by an edge,
  then an edge $\{x,y\}$ had to be removed by some two matched voters
  $v$ and $u$ such that $x \pref_v y$ and $y \pref_u x$.
  Since both voters belong to subelections $E'_1$ and $E'_2$, we obtain a contradiction that they are isomorphic via the given matchings.

  For the other direction, let us assume that $G$ has a size-$k$ clique.
  Let $K$ be the set of this clique's vertices.
  We form elections $E'_1$ and $E'_2$ by restricting $E_1$ and $E_2$ to the candidates from $K$.
  To verify that $E'_1$ and $E'_2$ are isomorphic via the given matchings, let us consider an arbitrary pair of matched voters $v$ and $u$ and arbitrary pair of candidates $x,y \in K$.
  It follows that $x \pref_v y$ and $x \pref_u y$.
  Otherwise an edge $\{x,y\}$ would be removed during the reduction, hence $K$ would not be a clique.
  A contradiction.
\end{proof}

\repeatproposition{Max-Cand-Both XP}

\begin{proof}
  Note that in the input of the problem we have the same number of candidates and the same number of voters in both elections.
  The algorithm simply guesses $k$ candidates from one of the input elections and checks if both obtained candidate-subelections restricted to chosen $k$ candidates are the same.
  There are $\binom{m}{k}$ many choices for $k$ candidates, hence the total running time is bounded by $O^*(m^k)$.

  Under ETH there is no $|V|^{o(K)}$-time algorithm for finding a size-$K$
  clique on a graph with vertex set $V$~\cite{ChenHKX06}.
  In our reduction in Theorem~\ref{thm:common-cand-both-np-hard} we have
  $m = |V|$ many candidates and required candidate set size is $k=K$.
  Therefore, using an $O^*(m^{o(k)})$-time algorithm for
  \textsc{Max. Common Cand.-Subelection with Both Matchings}
  we could solve \textsc{Clique} in time
  $m^{o(k)} \cdot \poly(m^2) \leq |V|^{o(k)}$. This contradicts ETH.
\end{proof}

\repeatproposition{Max-Cand-Both approximation}

\begin{proof}
  First, notice that we use a convention that an approximation ratio $\alpha$ is at least $1$ also for a maximization problem.
  Therefore, if {\sc Alg} is a value of a solution returned by an algorithm and {\sc Opt} is the value of an optimal solution then $\frac{\textsc{Opt}}{\textsc{Alg}} \leq \alpha$ implies that the solution is an $\alpha$-approximation of an optimal one.

  Let us describe the algorithm.
  For any fixed natural number $c \geq 1$ we check all size-$c$ subsets of candidates as a solution.
  There are at most $\binom{m}{c} \leq m^c$ many such subsets (i.e. polynomial many).
  We evaluate each of them in polynomial time.
  At least one of such trial is a subset of an optimal solution, so it is a feasible solution of size $c$.
  Therefore the algorithm has approximation ratio at most $\frac{\textsc{Opt}}{\textsc{Alg}} = \frac{t}{c}$.

  It is $\np$-hard to approximate \textsc{Clique} within factor $|V|^{1-\epsilon}$ for every $\epsilon > 0$, where $V$ is a given set of vertices~\cite{Zuckerman07}.
  In the reduction in Theorem~\ref{thm:common-cand-both-np-hard}
  the number of candidates $m$ equals to the number of vertices of the graph, hence $t \leq m = |V|$.
  Therefore, using an $t^{1-\epsilon}$-approximation algorithm
  for \textsc{Max. Common Cand.-Subelection with Both Matchings},
  for some $\epsilon > 0$, we could approximate \textsc{Clique} within
  $|V|^{1-\epsilon}$ factor. This implies $\p = \np$.
\end{proof}

As a comment we point that the above running time lower-bound (Proposition~\ref{prop:com-cand-both-xp}) and hardness of approximation (Proposition~\ref{prop:com-cand-both-apx}) can be combined under a stronger hypothesis.
The work of~\citet{ChalermsookCKLM17} give an evidence that for \textsc{Clique}, where $K$ denotes the size of maximum clique, $K^{1-\epsilon}$-approximation is not possible even in time FPT with respect to $K$ assuming gap version of ETH called Gap-ETH.
Hence, the same hardness holds to \textsc{Max. Common Cand.-Subelection with Both Matchings}.
It means, that the best way to solve both problems, even approximately, is to essentially enumerate all possibilities~\cite{ChalermsookCKLM17}.

\repeatproposition{Max. Common Cand.-Subelection problem with candidate matching}

\begin{proof}
  Below we give the reductions for all the three cases, i.e., the case
  with a given candidate matching, with a given voter matching, and
  without any matchings.

  \paragraph{The case with a given candidate matching.}
  We give a reduction from \textsc{Max. Common Cand.-Subelection with
  both Matchings}. Let $E_1 = (C,V)$ and $E_2 = (C,U)$ be our input
  elections, where $V = (v_1, \ldots, v_n)$ and
  $U = (u_1, \ldots, u_n)$, and let $t$ be the desired size of the
  isomorphic subelection (since we are in the setting with both
  matchings, we can assume that both elections are over the same
  candidate set). We assume that for each $i \in [n]$ voter $v_i$ is
  matched to $u_i$. Let $m = |C|$ and let $k = t/n$. We note that
  $k \leq m$.

  Our construction proceeds as follows. First, we form $m+1$ sets,
  $A$, $D_1, \ldots, D_m$, each containing $m+1$ new candidates. Let
  $\calD = A \cup D_1 \cup \cdots \cup D_m$. Note that
  $|\calD| = (m+1)^2$. We form elections $E'_1 = (C \cup \calD, V')$
  and $E'_2 = (C \cup \calD, U')$, where $V' = (v'_1, \ldots, v'_n)$
  and $U' = (u'_1, \ldots, u'_n)$. For each $i \in [n]$, we set their
  preference orders as follows (by writing $[v_i]$ or $[u_i]$ we mean
  copying the preference order of the respective voter):
  \begin{align*}
    v'_i \colon& D_1 \pref \cdots \pref D_{i-1} \pref A \pref D_{i} \pref \cdots \pref D_{m} \pref [v_i],\\
    u'_i \colon& D_1 \pref \cdots \pref D_{i-1} \pref A \pref D_{i} \pref \cdots \pref D_{m} \pref [u_i].
  \end{align*}
  Finally, we set the desired size of the isomorphic subelections
  to be $t' = n \cdot (k + (m+1)^2) = t + n(m+1)^2$.

  We claim that $E'_1$ and $E'_2$ have isomorphic candidate
  subelections of size $t'$ for the given candidate matching if and
  only if $E_1$ and $E_2$ have isomorphic candidate subelections of
  size $t$ for given candidate and voter matchings.

  Let us assume that $E'_1$ and $E'_2$ have the desired candidate
  subelections, $E''_1$ and $E''_2$. We claim that their isomorphism
  is witnessed by such a matching that for each $i$ voter $v'_i$ is
  matched to $u'_i$. If it were not the case, then to maintain the
  isomorphism these subelections would have to lose at least $m-1$
  candidates from $\calD$ (e.g., the candidates from $A$) and their
  sizes would be, at most, $n(m+m(m+1)) = n((m+1)^2-1) < t'$. Thus the
  isomorphism of $E''_1$ and $E''_2$ is witnessed by the same voter
  matching as the one required by our input instance. A simple
  counting argument shows that after dropping candidates from $\calD$
  from subelections $E''_1$ and $E''_2$, we obtain elections
  witnessing that $(E_1,E_2)$ is a \emph{yes}-instance of
  \textsc{Max. Common Cand.-Subelection with both Matchings}. The
  reverse direction is immediate.

  \paragraph{The case with a given voter matching.}
  This case follows by Proposition~\ref{prop:reduction} and the fact that
  \textsc{Cand.-Subelection with Voter Matching} is $\np$-complete.

  \paragraph{The case without any given matchings.}
  This case follows by Proposition~\ref{prop:reduction} and the fact
  that \textsc{Cand.-Subelection Isomorphism} problem is
  $\np$-complete.
\end{proof}

\repeatproposition{All 4 MCS}

\begin{proof}
  For the case without any matchings and the case with the voter
  matching, we use Proposition~\ref{prop:reduction} to reduce from the
  corresponding variant of \textsc{Subelection Isomorphism}.
  For the variants that include the candidate matching (for which
  \textsc{Subelection Isomorphism} is in $\p$), we reduce from the
  corresponding variants of \textsc{Max. Common Cand.-Subelection}.
  Let $E_1 = (C,V_1)$ and $E_2 = (C,V_2)$ be our input elections and
  let $t$ be the desired size of their isomorphic candidate
  subelections (since we are in the ``with candidate matching''
  regime, we take the candidate sets to be equal). Without loss of
  generality, we can assume that $|V_1| = |V_2|$; our
  $\np$-completeness proofs for \textsc{Max. Common Cand.-Subelection}
  give such instances.

  Let $m = |C|$, $n = |V_1| = |V_2|$, and let $D$ be a set of $(n-1)m$
  dummy candidates. We form elections $E'_1$ and $E'_2$ to be
  identical to $E_1$ and $E_2$, respectively, except that they also
  include the candidates from $D$, who are always ranked on the
  bottom, in the same order.
  Hence, the number of candidates in $E'_1$ and $E'_2$ equals $nm$.
  We ask if $E'_1$ and $E'_2$ have
  isomorphic subelections of size $t' = t + n(n-1)m$.

  If $E_1$ and $E_2$ have isomorphic candidate subelections of size
  $t$, then certainly $E'_1$ and $E'_2$ have isomorphic subelections
  of size $t'$ (it suffices to take the same subelections as for
  $E_1$ and $E_2$ and include the candidates from $D$).

  On the other hand,
  if $E'_1$ and $E'_2$ have isomorphic subelections of size
  $t'$, then $E_1$ and $E_2$ have size-$t$ isomorphic candidate
  subelections. Indeed, the subelections of $E'_1$ and $E'_2$ must
  include all the $n$ voters. Otherwise their sizes would
  be at most $(n-1)mn < t+(n-1)mn \leq t'$.
  Thus the subelections of $E'_1$ and $E'_2$ are candidate
  subelections. As we can also assume that the subelections of
  $E'_1$ and $E'_2$ include all the candidates from $D$, by omitting
  these candidates we get the desired candidate subelections of $E_1$
  and~$E_2$.
\end{proof}

\section{ILP for Maximum Common Voter-Subelection}\label{appendix:ilp}
Since the polynomial-time algorithm for \textsc{Max. Common
Voter-Subelection} is fairly slow, we express the problem as an
ILP. We have two types of variables:
\begin{enumerate}
\item For each pair of voters $v \in V$ and $u \in U$, we have a binary
  variable $N_{v,u}$. If it is set to $1$, then we interpret it as
  saying that voter $v$ is included in the subelection of $E$,
  voter $u$ is included in the subelection of $F$, and the two
  voters are matched. Value~$0$ means that the preceding statement
  does not hold.
\item For each pair of candidates $c \in C$ and $d \in D$, we have a
  binary variable $M_{c,d}$. If it is set to $1$ then we interpret it
  as saying that $c$ is matched to $d$ in the isomorphic subelections
  (note that, since we are looking for voter subelections, every
  candidate from $C$ has to be matched to some candidate from $D$, and
  the other way round).
\end{enumerate}
To ensure that variables $N_{v,u}$ and $M_{c,d}$ describe the
respective matchings, we have the following basic constraints:
\begin{align*}
  &\textstyle \sum_{u \in U}N_{v,u} \leq 1, \;\; \forall v \in V,&
  &\textstyle \sum_{d \in D}M_{c,d} = 1, \;\; \forall c \in C, \\
  &\textstyle \sum_{v \in V}N_{v,u} \leq 1, \;\; \forall u \in U,&
  &\textstyle \sum_{c \in C}M_{c,d} = 1, \;\; \forall d \in D.
\end{align*}
For each pair of voters $v \in V$, $u \in U$ and each pair of
candidates $c \in C$ and $d \in D$, we introduce constant
$w_{v,u,c,d}$ which is set to $1$ if $v$ ranks $c$ on the same
position as $u$ ranks $d$, and which is set to $0$ otherwise. We use
these constants to ensure that the matchings specified by variables
$N_{v,u}$ and $M_{c,d}$ indeed describe isomorphic
subelections. Specifically, we have the following constraints (let
$m = |C| = |D|$):
\begin{align*}
  \textstyle \sum_{c \in C}\sum_{d \in D} w_{v,u,c,d}\cdot M_{c,d} \geq m \cdot N_{v,u}, \quad \forall v \in V, u \in U.
\end{align*}
For each $v \in V$ and $u \in U$, they ensure that if $v$ is matched
to $u$ then each candidate $c$ appears in $v$ on the same position as
the candidate matched to $c$ appears in $u$.

\end{document}